\documentclass{mcom-l}
\numberwithin{equation}{section}
\newtheorem{theorem}{Theorem}[section]

\theoremstyle{definition}

\theoremstyle{remark}
\newtheorem*{remark}{Remarks}
\theoremstyle{openproblem}
\newtheorem*{openproblem}{Open Problem}

\begin{document}

\bibliographystyle{amsplain}

\title{Two Compact Incremental Prime Sieves}

\author[J.~P.~Sorenson]{Jonathan P.~Sorenson}
\address{Computer Science and Software Engineering,
Butler University, Indianapolis, IN 46208 USA}
\email{sorenson@butler.edu}
\urladdr{http://www.butler.edu/~sorenson} 
\thanks{Supported by a grant from the Holcomb Awards Committee}

\subjclass[2000]{Primary 11Y16, 68Q25; Secondary 11Y11, 11A51}
\keywords{Prime number sieve, number thoeretic algorithms,
  algorithm analysis, primality testing}
\date{\today}

\begin{abstract}
A prime sieve is an algorithm that finds the primes up to a bound $n$.
We say that a prime sieve is
  \textit{incremental}, if it can quickly determine if $n+1$ is
  prime after having found all primes up to $n$.
We say a sieve is \textit{compact}
  if it uses roughly $\sqrt{n}$ space or less.
In this paper we present two new results:
\begin{itemize}
  \item We describe the \textit{rolling sieve}, a practical, incremental
    prime sieve that takes $O(n\log\log n)$ time and
    $O(\sqrt{n}\log n)$ bits of space, and
  \item We show how to modify the sieve of Atkin and Bernstein \cite{AB2004}
    to obtain a sieve that is simultaneously 
    sublinear, compact, and incremental.
\end{itemize}
The second result solves an open problem given by 
  Paul Pritchard in 1994 \cite{Pritchard94}.
\end{abstract}

\maketitle

\section{Introduction and Definitions}

A \textit{prime sieve} is an algorithm that finds all prime numbers
  up to a given bound $n$.
The fastest known algorithms, 
  including Pritchard's wheel sieve \cite{Pritchard81}
  and the Atkin-Bernstein sieve \cite{AB2004},
  can do this using at most $O(n/\log\log n)$ arithmetic operations.
The easy-to-code sieve of Eratosthenes requires $O(n\log\log n)$ time,
  and there are a number of sieves in the literature that require
  linear time \cite{Pritchard83,Pritchard87}.

Normally, running time is the main concern in algorithm design, but
  in this paper we are also interested in two other properties:
  \textit{incrementality} and \textit{compactness}.

We say that a sieve is \textit{compact} 
  if it uses at most $n^{1/2+o(1)}$ space.
Bays and Hudson \cite{BH77} showed how to segment the sieve of Eratosthenes
  so that only $O(\sqrt{n})$ bits of space are needed.
  (See also \cite{Brent73}.)
Pritchard \cite{Pritchard83} showed how to apply a fixed wheel to
  get a compact sieve that runs in linear time.
Atkin and Bernstein \cite{AB2004} gave the first compact sieve that
  runs in sublinear time.
For other sieves that address space issues, see for example
  \cite{Galway2000,Galway98,GalwayThesis,Sorenson98,DJS96,Sorenson06}.

Loosely speaking, a sieve is \textit{incremental} if it can
  determine the primality of $n+1$ after having found all primes up to $n$
  using a small amount of additional work.
Bengelloun \cite{Bengelloun86} presented the first prime sieve that needed only
  a constant amount of additional work to do this, thereby taking a total of
  $O(n)$ time to find all primes up to $n$.
Pritchard \cite{Pritchard94} showed how to improve Bengelloun's sieve using a
  dynamic wheel so that in constant time it determines the
  the primality of all integers from $n$ to $n+\Theta(\log\log n)$.
In other words, it takes
  $O(1+ (p-n) / \log \log n)$ time to find $p$, if $p$ is the smallest prime
  exceeding $n$.
Pritchard's sieve takes $O(n/\log\log n)$ time to find all primes up to $n$.

To clarify, let us define a sieve as \textit{$t(n)$-incremental} if,
  in the worst case, it requires
  $(p-n)\cdot t(n)+O(1)$ operations to find $p = p_{\pi(n)+1}$,
  the smallest prime exceeding $n$, 
  after having found all primes up to $n$.
Thus, Bengelloun's sieve is $O(1)$-incremental, and Pritchard's sieve
  is $O(1/\log\log n)$-incremental.

\begin{openproblem}
Design a prime sieve that is both compact and $o(1)$-incremental.
\cite{Pritchard94}
\end{openproblem}

\noindent
In this paper, we address this problem from two perspectives,
one practical, and one theoretical:
\begin{itemize}
\item
We present a sieve algorithm, with pseudocode,
  that finds the primes up to $n$
  using $O(n\log\log n)$ arithmetic operations and
  $O(\sqrt{n}\log n)$ bits of space, 
  that is  $O(\log n/\log\log n)$-incremental.
It is based on the segmented sieve of Eratosthenes, adapting some
  of the ideas in Bennion's sieve \cite{Galway98}.
Our sieve uses a circular array of stack datastructures
  and works well in practice.
\item
We also prove the following:
\begin{theorem}\label{mainthm}
There exists a prime sieve that is
  both compact and $O(1/\log\log n)$-incremental.
\end{theorem}
\noindent
Our proof relies on modifying the sieve of Atkin and Bernstein \cite{AB2004}.
\end{itemize}

After we discuss some preliminaries in \S\ref{sec:prelim},
  we present our rolling sieve in \S\ref{sec:alg} and
  we prove our theorem in \S\ref{sec:thm}.

\section{Preliminaries\label{sec:prelim}}

In this section we discuss our model of computation,
  review some helpful estimates from elementary number theory,
  and review the sieve of Eratosthenes.

\subsection{Model of Computation}

Our model of computation is a RAM with a potentially infinite,
direct access memory.

If $n$ is the input, then all arithmetic operations on integers
  of $O(\log n)$ bits have unit cost.
This includes $+$, $-$, $\times$, and division with remainder.
Comparisons, array indexing, assignment, branching,
bit operations, and other basic operations are also assigned unit cost.
Memory may be addressed at the bit level or at the word level, where
  each word has $O(\log n)$ bits.

Space is measured in bits.
Thus, it is possible for an algorithm to touch $n$ bits in only
  $O(n/\log n)$ time if memory is accessed at the word level.
The space used by the output of a prime sieve,
  the list of primes up to $n$, is not counted against the algorithm.

This is the same model used in \cite{DJS96,Sorenson98,Sorenson06}.

\subsection{Some Number Theory}

We make use of the following estimates.
Here the sums over $p$ are over primes only, and $x>0$:
\begin{eqnarray}
\pi(x) := \sum_{p\le x} 1 &=& \frac{x}{\log x} (1+o(1));
  \label{eq:pi} \\
\sum_{p\le x} \frac{1}{p} &=& \log\log x + O(1);
  \label{eq:harmonic} \\
\sum_{p\le x} \log p &=& x(1+o(1));
  \label{eq:logp} 
\end{eqnarray}
For proofs, see Hardy and Wright \cite{HW}.

\subsection{Sieve of Eratosthenes\label{sec:erat}}

We present algorithms using a C++ style that should be
  familiar to most readers.
We assume that integer variables can accurately hold
  any integer represented in standard signed binary
  notation in $O(\log n)$ bits.
In practice, this means 64 bits.

Recall that the sieve of Eratosthenes uses a bit vector
  to represent the primes up to $n$.
All are assumed to be prime (aside from 0 and 1) and,
  for each prime $p$ found, its multiples $q$ are enumerated
  and "crossed off" as not prime, as follows:

\begin{verbatim}
  BitVector S(n+1); // bit vector for 0..n
  S.setAll();
  S[0]=0; S[1]=0;
  for(p=2; p*p<=n; p=p+1)
    if(S[p]==1)  // so that p must be prime
      for(q=2*p; q<=n; q=q+p)
        S[q]=0;
\end{verbatim}
This requires $\sum_{p\le \sqrt{n}} n/p=O(n\log\log n)$ arithmetic
  operations (using (\ref{eq:harmonic})) and $O(n)$ bits of space.

\subsection{Segmented Sieve of Eratosthenes\label{sec:seg}}

The main loop of a segmented sieve is simply a loop over each segment.
Here $\Delta$ is the size of that segment; we choose
  $\Delta$ to be proportional to $\sqrt{n}$.
We assume here the primes below $\sqrt{n}$ have 
already been found some other way (perhaps with the unsegmented sieve above).

\begin{verbatim}
  for(left=sqrt(n)+1; left<=n; left=left+delta)
  {
    right=min(left+delta-1,n);
    sieve(left,right);
  }
\end{verbatim}
Each segment is then sieved by crossing off multiples of each prime
  $p$ below $\sqrt{n}$.

\begin{verbatim}
  Primelist P(sqrt(right)); // list of primes <= sqrt(n)
  BitVector S(delta); // bit vector for left..right
  S.setAll();
  for(i=0; i<P.length; i++)
  {
    p=P[i];
    first=left+(p-(left%p))%p; // min. first>=left st. p|first
    for(q=first; q<=right; q=q+p)
      S[q-left]=0;
  }
\end{verbatim}
Using (\ref{eq:harmonic}), 
the running time for a single segment is proportional to
$$
  \sum_{p\le\sqrt{n}} \left( 1+\frac{\Delta}{p} \right)
   \quad=\quad O\left(\Delta\log\log n + \frac{\sqrt{n}}{\log n} \right)
$$
and summing this cost over $n/\Delta$ segments gives $O(n\log\log n)$.
The total space used is $O(\sqrt{n})$ bits;
  this is dominated by space for the the list of
  primes up to $\sqrt{n}$  (using (\ref{eq:logp})), 
  plus $\Delta$ bits for the segment bit vector.

\begin{remark}\ 
\begin{itemize}
\item
  Both sieves can readily be modified to completely factor all integers
    up to $n$ by replacing the bit vector with an array of lists,
    and adding $p$ to the list for each $q$ generated by the inner loop.
\item
  Pritchard\cite{Pritchard83} added a \textit{static wheel} to
    the segmented sieve to obtain an $O(n)$ running time.
  For a discussion of the static wheel, see \cite[\S2.4]{Sorenson98},
    which gives C++-style pseudocode and a running time analysis.
\item
  For a parallel version of the segmented sieve, see \cite{SP94}.
\end{itemize}
\end{remark}

\section{Algorithm Description\label{sec:alg}}

In this section we present our rolling sieve.
It is in many ways similar to the hopping sieve of Bennion \cite{Galway98}.

The primary data structure is a circular array of stacks.
Each array location corresponds to an integer $n$,
  and when $n$ is reached by the sieve,
  its stack will contain the prime divisors of $n$ up to $\sqrt{n}$.
If the stack is empty, then either $n$ is prime or the square of a prime.
In the latter case, we discover a new prime by which to sieve, $r$,
  so we push $r$ onto $n+r$'s stack.

When moving from $n$ to $n+1$, we pop each prime $p$ from $n$'s stack
  and push it onto the stack for $n+p$, the next stack position
  corresponding to an integer divisible by $p$.

If the position for $n+p$ is larger than $\Delta$, 
  the size of the circular array,
  we simply wrap around to the front.

We are now ready to present pseudocode.

\subsection{Precomputation}

Let \texttt{start} be the value of the first integer we wish to
  test for primality.
To set up, then, we find the primes up to $\sqrt{\texttt{start}}$
  and push them onto the correct stacks.
Note that array position $0$ corresponds to \texttt{start} as we begin.

\begin{verbatim}
  r=floor(sqrt(start))+1;
  s=r*r;
  PrimeList P(r-1);
  delta=r+2;
  StackArray T(delta);
  for(i=0; i<P.length; i=i+1)
  {
    p=P[i];
    j=(p-(start%p))%p;
    T[j].push(p);
  }
  pos=0; n=start;
\end{verbatim}
We have $\lfloor \sqrt{n} \rfloor = r-1 = \Delta-3$ so that
  $n< s=r^2$ and $r+1<\Delta$ hold.

\subsection{Primality of $n$}

Once the array of stacks is set up, we can check each successive
  integer $n$ for primality as shown in the function below,
  which returns true if $n$ is prime, and false otherwise.
It also sets its state to be ready to handle $n+1$ on the subsequent call.

\begin{verbatim}
bool next()
{
  isPrime=true;
  while(!T[pos].isEmpty()) // process prime divisors
  {
    p=T[pos].pop();
    T[(pos+p)%delta].push(p);
    isPrime=false;
  }
  if(n==s) // then n is a square
  {
    if(isPrime) // then r is in fact prime
    {
      T[(pos+r)%delta].push(r);
      isPrime=false;
    }
    r=r+1; s=r*r;
  }
  n=n+1;
  pos=(pos+1)%delta;
  if(pos==0) { delta=delta+2; }
  return isPrime;
}
\end{verbatim}
\begin{itemize}
\item The list of primes stored in stacks includes all primes
  up to $\sqrt{n}$, unless $n$ is, in fact, the square of a prime.
  If this is the case, we detect that $r$ is prime and
  put it into the appropriate stack.
  This is how we incrementally grow the list of primes up to $\sqrt{n}$
    for sieving.

\item We can grow $\Delta$ and the array of stacks over time by
  simply adding two to $\Delta$ when \texttt{pos} reaches zero.
  We make sure no prime stored in the stacks exceeds $\Delta$,
    or in other words, $r<\Delta$ is invariant.
  The new stacks added to the end of the array are initialized to empty.

  Note that if we start $\Delta=\lfloor\sqrt{n}\rfloor+3$,
    and then after $\Delta$ function calls increment by 2,
    we end up iterating the mapping 
    $(n,\Delta)\rightarrow(n+\Delta,\Delta+2)$.
  Iterating $k$ times gives 
    $(n,\Delta)\Rightarrow(n+k\Delta+k(k-1),\Delta+2k)$.
  Squaring $\Delta+2k$ shows that the segment stays larger than $\sqrt{n}$
    over time, and the $k(k-1)$ term insures it won't exceed $O(\sqrt{n})$.

  Extending an array of stacks in the RAM model discussed 
    in \S\ref{sec:prelim} is straightforward,
    but in practice it is not so simple.
  One option is to use a fixed value for $\Delta$,
    only use primes up to $\Delta$ in the stacks,
    and then numbers that seem prime should be prime tested
    as an extra step.
  See, for example, the ideas in \cite{Sorenson06}.

\item To find all primes up to a bound $n$,
  simply print the primes up to $100$, say, 
  then set up using $\texttt{start}=100$,
  and repeatedly call the \texttt{next()} function, printing
  when it returns true.

  \begin{verbatim}
    int nextprime()
      { while(!next()); return n-1; }
  \end{verbatim}

\item
  The sieve can readily be modified to generate integers in factored form.
  Simply make a copy of $n$, and as the primes are popped from the stack,
    divide them into $n$.
  When the stack is empty, what remains is either 1 or prime.
  We leave the details to the reader.

  An early version of this algorithm was used for this purpose
    in generating data published in \cite{BS13}.

\item
  The array of stacks datastructure used here can easily be viewed
  as a special-case application of a \textit{monotone priority queue}.
  In this view, each prime up to $\sqrt{n}$ is stored in the
    queue with its priority set to the next multiple of that prime.
  See, for example, \cite{CGS99}.
\end{itemize}

\subsection{Analysis}

\begin{theorem}
  The rolling sieve will find all primes up to a bound $n$
  using $O(n\log\log n)$ arithmetic operations and
  $O(\sqrt{n}\log n)$ bits of space.
\end{theorem}

\begin{proof}
The running time is bounded by the number of times
  each prime $p$ up to $\sqrt{n}$ is popped and pushed.
But that happens exactly once for each integer up to $n$
  that is divisible by $p$, or
$$
  \sum_{p\le\sqrt{n}} \frac{n}{p} = O( n\log\log n).
$$
Assuming a linked list style stack data structure,
  the total space used is one machine word for the number of stacks, 
  $\Delta=O(\sqrt{n})$,
  plus the number of nodes, $\pi(\sqrt{n})$.
At $O(\log n)$ bits per word, we have our result.
\end{proof}

\begin{theorem}
  The rolling sieve is $O(\log n/\log\log n)$-incremental.
\end{theorem}

\begin{proof}
It should be clear that we need to count the total number of
  prime divisors of the integers from $n$ to $p$, where $p$ is
  the smallest prime larger than $n$.

Let $\ell=p-n$.
First we consider primes $q\le\ell$.
Each such prime $q$ divides $O(\ell/q)$ integers in this interval,
  for a total of $O(\ell\log\log \ell)$.

Next we consider primes $q>\ell$.
Each such prime can divide at most one integer between $n$ and $p$.
Summing up the logarithms of such primes, we have
$$
  \sum_{m=n}^p \sum_{q\mid m, q>\ell} \log q = O(\ell \log n).
$$
The number of such primes $q$ is maximized when the primes are as small
  as possible.
Again, each prime can only appear once in the sum above,
  so we solve the following for $x$:
$$
  \sum_{\ell<q \le x} \log q = O(\ell\log n)
$$
We have $x=O(\ell\log n)$ by (\ref{eq:logp}), 
  and so the number of primes is bounded by
  $\pi(x)=O(\ell\log n/(\log\ell+\log\log n))$.

Dividing through by $\ell$ completes the proof.
\end{proof}

\begin{remark}\ 
\begin{itemize}
\item
  For the purposes of analysis, and to limit the space used,
  our stacks are linked-list based.
  For speed, array-based stacks are better. 
  For 64-bit integers, an array of length 16 is sufficient.
\item
  Although we can only prove the rolling sieve is
  $O(\log n/\log\log n)$-incremental, 
  by the Erd\"os-Kac theorem, in practice it will behave
  as if it is $O(\log\log n)$-incremental.
  See, for example, \cite{EK40}.
\item
  One could add a wheel to reduce the overall running time by
  a factor proportional to $\log\log n$,
  and similarly improve incremental behavior.
  We have not tried to code this, 
    and growing the wheel by adding even
    a single prime, in an incremental way, seems messy at best.
  From a theoretical perspective, the question is moot as we will see
    in the next section.
\end{itemize}
\end{remark}

\section{The Theoretical Solution\label{sec:thm}}

In this section we prove Theorem \ref{mainthm}.

Our proof makes use of the sieve of Atkin and Bernstein \cite{AB2004}.
In particular, we make use of the following properties of this
  algorithm:
\begin{itemize}
  \item The sieve works on segments of size $\Delta$, where
    $\Delta\approx\sqrt{n}$, and has a main loop similar to
    the segmented sieve of Eratosthenes discussed in \S\ref{sec:seg}.
  \item The time to determine the primality of all integers
    in the interval from $n$ to $n+\Delta$ is
    $O(\Delta/\log\log n)$ arithmetic operations.
  \item Each interval can be handled separately, with at most
    $O(\Delta)$ extra space of data passed from one interval to the next.
\end{itemize}
Any sieve that had these properties can be used as a basis for our proof.

\begin{proof}
We maintain two consecutive intervals of information.

The first interval is a simple bit vector and list of primes
  in the interval $n$ to $n+\Delta$.
Calls to the \texttt{next()} function are answered using data
  from this interval in constant time.
We can also handle a \texttt{nextprime()} function that jumps
  ahead to the next prime in this interval in constant time
  using the list.

The second interval, for $n+\Delta$ to $n+2\Delta$, is in process; 
  it is being worked on by the Atkin-Bernstein sieve.
We maintain information that allows us to start and stop the 
  sieve computation after any fixed number of instructions.
When this interval is finished, it creates a bit vector and list of primes
  ready for use as the new first interval.

So, when a call to the \texttt{next()} function occurs,
  after computing its answer from the first interval,
  a constant amount of additional work time is invested sieving the second
  interval; enough that after $\Delta$ calls to \texttt{next()},
  the second interval will be completed.

When a call to the \texttt{nextprime()} function occurs,
  we compute the distance to that next prime, $\ell$,
  and invest $O(1+\ell/\log\log n)$ time sieving the second interval
  after computing the function value.

In this way, by the time the first interval has been used up,
  the second interval has been completely processed, and we can
  replace the first interval with the second, and start up the next one.
\end{proof}

\begin{remark}\ 
\begin{itemize}
\item
The proof is a bit artificial and of theoretical interest only.
Coding this in practice seems daunting.
\item
Any segmented sieve, more or less, can be used to make this proof work.
In particular, if a compact sieve that is faster, or uses less space,
can be found, such a sieve immediately implies a version that is
also incremental.
\item
Pritchard has also pointed out that \textit{additive} sieves,
  ones that avoid the use of multiplication and division operations,
  are desirable due to the fact that hardware addition is much faster
  than multiplication or division.
From a theoretical perspective, this is moot since a multiplication table
  together with finite-precision multiplication and division routines
  can reduce all arithmetic operations on $O(\log n)$-bit integers
  to additions.

That said, see \cite{Pritchard94,Pritchard81}.
\end{itemize}
\end{remark}


\bibliography{all}

\begin{thebibliography}{10}

\bibitem{AB2004}
A.~O.~L. Atkin and D.~J. Bernstein, \emph{Prime sieves using binary quadratic
  forms}, Mathematics of Computation \textbf{73} (2004), 1023--1030.

\bibitem{BS}
Eric Bach and Jeffrey~O. Shallit, \emph{Algorithmic number theory}, vol.~1, MIT
  Press, 1996.

\bibitem{BS13}
Eric Bach and Jonathan~P. Sorenson, \emph{Approximately counting semismooth
  integers}, Proceedings of the 38th International symposium on symbolic and
  algebraic computation (New York, NY, USA), ISSAC '13, ACM, 2013, pp.~23--30.

\bibitem{BH77}
C.~Bays and R.~Hudson, \emph{The segmented sieve of {E}ratosthenes and primes
  in arithmetic progressions to $10^{12}$}, BIT \textbf{17} (1977), 121--127.

\bibitem{Bengelloun86}
S.~Bengelloun, \emph{An incremental primal sieve}, Acta Informatica \textbf{23}
  (1986), no.~2, 119--125.

\bibitem{Brent73}
R.~P. Brent, \emph{The first occurrence of large gaps between successive
  primes}, Mathematics of Computation \textbf{27} (1973), no.~124, 959--963.

\bibitem{CGS99}
Boris~V. Cherkassky, Andrew~V. Goldberg, and Craig Silverstein, \emph{Buckets,
  heaps, lists, and monotone priority queues}, SIAM J. Comput. \textbf{28}
  (1999), no.~4, 1326–--1346, http://dx.doi.org/10.1137/S0097539796313490.

\bibitem{CP}
R.~Crandall and C.~Pomerance, \emph{Prime numbers, a computational
  perspective}, Springer, 2001.

\bibitem{DJS96}
Brian Dunten, Julie Jones, and Jonathan~P. Sorenson, \emph{A space-efficient
  fast prime number sieve}, Information Processing Letters \textbf{59} (1996),
  79--84.

\bibitem{EK40}
P.~Erd\"os and M.~Kac, \emph{The gaussian law of errors in the theory of
  additive number theoretic functions}, American Journal of Mathematics
  \textbf{62} (1940), no.~1, 738--742.

\bibitem{Galway98}
William~F. Galway, \emph{Robert bennion's ``hopping sieve''}, Proceedings of
  the Third International Symposium on Algorithmic Number Theory (London, UK,
  UK), ANTS-III, Springer-Verlag, 1998, pp.~169--178.

\bibitem{Galway2000}
\bysame, \emph{Dissecting a sieve to cut its need for space}, Algorithmic
  number theory (Leiden, 2000), Lecture Notes in Comput. Sci., vol. 1838,
  Springer, Berlin, 2000, pp.~297--312. \MR{2002g:11176}

\bibitem{GalwayThesis}
\bysame, \emph{Analytic computation of the prime-counting function}, Ph.D.
  thesis, University of Illinois at Urbana-Champaign, 2004, Available at
  http://www.math.uiuc.edu/\verb!~!galway/PhD\verb!_!Thesis/.

\bibitem{HW}
G.~H. Hardy and E.~M. Wright, \emph{An introduction to the theory of numbers},
  5th ed., Oxford University Press, 1979.

\bibitem{Knuthv2}
D.~E. Knuth, \emph{The art of computer programming: Seminumerical algorithms},
  3rd ed., vol.~2, Addison-Wesley, Reading, Mass., 1998.

\bibitem{Pritchard81}
P.~Pritchard, \emph{A sublinear additive sieve for finding prime numbers},
  Communications of the ACM \textbf{24} (1981), no.~1, 18--23,772.

\bibitem{Pritchard83}
\bysame, \emph{Fast compact prime number sieves (among others)}, Journal of
  Algorithms \textbf{4} (1983), 332--344.

\bibitem{Pritchard87}
\bysame, \emph{Linear prime-number sieves: A family tree}, Science of Computer
  Programming \textbf{9} (1987), 17--35.

\bibitem{Pritchard94}
\bysame, \emph{Improved incremental prime number sieves}, First International
  Algorithmic Number Theory Symposium (ANTS-1) (L.~M. Adleman and M.-D. Huang,
  eds.), May 1994, LNCS 877, pp.~280--288.

\bibitem{Sorenson98}
Jonathan~P. Sorenson, \emph{Trading time for space in prime number sieves},
  Proceedings of the Third International Algorithmic Number Theory Symposium
  (ANTS III) (Portland, Oregon) (Joe Buhler, ed.), 1998, LNCS 1423,
  pp.~179--195.

\bibitem{Sorenson06}
\bysame, \emph{The pseudosquares prime sieve}, Proceedings of the 7th
  International Symposium on Algorithmic Number Theory (ANTS-VII) (Berlin,
  Germany) (Florian Hess, Sebastian Pauli, and Michael Pohst, eds.), Springer,
  July 2006, LNCS 4076, ISBN 3-540-36075-1, pp.~193--207.

\bibitem{SP94}
Jonathan~P. Sorenson and Ian Parberry, \emph{Two fast parallel prime number
  sieves}, Information and Computation \textbf{144} (1994), no.~1, 115--130.

\end{thebibliography}

\nocite{Pritchard81,Pritchard83,Pritchard94,Bengelloun86,
  HW,BS,CP,Knuthv2,Galway2000,GalwayThesis,
  Sorenson06,SP94,Sorenson98,DJS96,AB2004}
  
\providecommand{\bysame}{\leavevmode\hbox to3em{\hrulefill}\thinspace}
\providecommand{\MR}{\relax\ifhmode\unskip\space\fi MR }
\providecommand{\MRhref}[2]{%
  \href{http://www.ams.org/mathscinet-getitem?mr=#1}{#2}
}
\providecommand{\href}[2]{#2}

\end{document}